\documentclass[12pt]{iopart}
\usepackage{iopams}
\usepackage{amsthm}
\usepackage{graphicx}
\newtheorem{Thm}{Theorem}
\newtheorem{Lem}[Thm]{Lemma}
\theoremstyle{definition}

\newcommand{\bra}[1]{{\left\langle #1 \right|}}
\newcommand{\ket}[1]{{\left| #1 \right\rangle}}

\newcommand{\T}{\mbox{$\mathrm{tr}$}}

\begin{document}
\title{Unified entropy, entanglement measures and monogamy of multi-party entanglement}

\author{Jeong San Kim and Barry C. Sanders }

\address{
 Institute for Quantum Information Science,
 University of Calgary, Alberta T2N 1N4, Canada
} \eads {\mailto{jekim@ucalgary.ca}
}

\date{\today}
\begin{abstract}
We show that restricted shareability of multi-qubit entanglement can
be fully characterized by unified-$(q,s)$ entropy. We provide a
two-parameter class of bipartite entanglement measures, namely
unified-$(q,s)$ entanglement with its analytic formula in two-qubit
systems for $q\geq 1$, $0\leq s \leq1$ and $qs\leq3$. Using
unified-$(q,s)$ entanglement, we establish a broad class of the
monogamy inequalities of multi-qubit entanglement for $q\geq2$,
$0\leq s \leq1$ and $qs\leq3$.
\end{abstract}
\maketitle

\section{Introduction}
Quantum entanglement is a physical resource with various
applications to quantum information and communication processing.
Quantum teleportation uses maximal entanglement between two
particles as a resource to transfer an unknown quantum state from
one particle to another without sending the actual particle itself
~\cite{tele}. The non-local correlation of quantum entanglement also
provides us with secure cryptographic keys~\cite{qkd1, qkd2}.

Whereas classical correlation can be freely shared among parties in
multi-party systems, quantum entanglement is restricted in its
shareability. If a pair of parties are maximally entangled in
multipartite systems, they cannot share entanglement~\cite{ckw,ov}
nor classical correlations~\cite{kw} with the rest of the system,
thus the term {\em monogamy of entanglement}~(MoE)~\cite{T04}.

MoE lies at the heart of many quantum information and communication
protocols. In quantum cryptography, for example, MoE is
fundamentally important because it quantifies how much information
an eavesdropper could potentially obtain about the secret key to be
extracted; the founding principle of quantum cryptographic schemes
that an eavesdropper cannot obtain any information without
disturbance is guaranteed the law of quantum physics, namely MoE
rather than assumptions on the difficulty of computation.

The first mathematical characterization of MoE was established for
three-qubit systems as an inequality in terms of
concurrence~\cite{ww}, which is referred as the
Coffman-Kundu-Wootters (CKW) inequality~\cite{ckw}. The CKW
inequality was generalized for multi-qubit systems~\cite{ov} and for
some cases of multi-qudit systems~\cite{kds}, and dual monogamy
inequalities were also proposed for multi-party quantum
systems~\cite{gbs, bgk, kpoly}.

However, there exist quantum states in higher-dimensional systems
violating CKW inequality~\cite{ou, ks}; thus the CKW inequality in
multi-qubit systems fails in its generalization into
higher-dimensional quantum systems. Moreover, characterizing MoE as
an inequality is not generally true for other entanglement measures
such as {\em entanglement of formation}~(EoF)~\cite{bdsw}; monogamy
inequality in terms of EoF is not valid even in multi-qubit systems.
Thus it is important to have a proper entanglement measures to
characterize MoE not only for the study of general MoE in
higher-dimensional quantum systems but in multi-qubit systems as
well.

The proof of the CKW inequality in multi-qubit systems~\cite{ckw} is
based on the feasibility of analytic evaluation of concurrence for
two-qubit mixed states. In fact, there are various possible
definitions of bipartite entanglement measure using different
entropy functions such as R\'enyi-$\alpha$ and Tsallis-$q$
entropies~\cite{renyi, horo, tsallis, lv}. For selective ranges of
$\alpha$ and $q$, these entanglement measures are tractable in
two-qubit systems, and, moreover, monogamy inequality of multi-qubit
entanglement is feasible in terms of these measures~\cite{ks2, KT}.

Here we establish a unification of monogamy inequalities in
multi-qubit systems. Using unified-$(q,s)$ entropy with real
parameters $q$ and $s$~\cite{ue1, ue2}, we define a class of
bipartite entanglement measures namely {\em unified-$(q,s)$
entanglement}, and show a broad class of monogamy inequalities of
multi-qubit systems in terms of unified-$(q,s)$ entanglement.

Our result shows that unified-$(q,s)$ entanglement contains
concurrence, EoF, R\'enyi-$\alpha$ and Tsallis-$q$ entanglement as
special cases, showing their explicit relation with respect to a
smooth function. Furthermore, our result reduces to every known case
of multi-qubit monogamy inequalities such as R\'enyi and Tsallis
monogamy~\cite{ks2, KT} and the CKW inequality for selective choices
of $q$ and $s$. Thus, our result provides an interpolation of the
previous results about monogamy of multi-qubit entanglement.

This paper is organized as follows. In Section~\ref{Subsec:
definition}, we define unified-$(q,s)$ entanglement for bipartite
quantum states, and provide its relation with concurrence, EoF,
R\'enyi-$q$ entanglement and Tsallis-$q$ entanglement. In
Section~\ref{Subsec: 2formula}, we provide an analytic formula of
unified-$(q,s)$ entanglement in two-qubit systems for $q \geq 1$, $0
\leq s \leq1$ and $qs\leq3$. In Section~\ref{Sec: monopoly}, we
derive a monogamy inequality of multi-qubit entanglement in terms of
unified-$(q,s)$ entanglement for $q \geq 2$, $0 \leq s \leq1$ and
$qs\leq3$. We summarize our results in Section~\ref{Conclusion}.


\section{Unified-$(q,s)$ Entanglement}
\label{Sec: Tqentanglement}

\subsection{Definition}
\label{Subsec: definition}

For a quantum state $\rho$, unified-$(q,s)$ entropy is
\begin{equation}
S_{q,s}(\rho):=\frac{1}{(1-q)s}\left[{\left(\T \rho^{q}\right)}^s-1\right],
\label{uqs-entropy}
\end{equation}
for $q,~s \geq 0$ such that $q \neq 1$ and $s \neq 0$.
Unified-$(q,s)$ entropy converges to
R\'enyi-$q$ entropy~\cite{horo},
\begin{equation}
\lim_{s \rightarrow 0}S_{q,s}(\rho)=\frac{1}{1-q}\log \T \rho^{q}=R_{q}(\rho),
\label{Renyi}
\end{equation}
and also tends to Tsallis-$q$ entropy~\cite{lv},
\begin{equation}
\lim_{s \rightarrow 1}S_{q,s}(\rho)=\frac{1}{1-q}\left(\T \rho^{q}-1\right)=T_{q}(\rho).
\label{Tsallis}
\end{equation}
For the case that $q$ tends to
1, $S_{q,s}(\rho)$ converges to the von Neumann entropy, that is
\begin{eqnarray}
\lim_{q \rightarrow 1}S_{q,s}(\rho)=-\T \rho\log\rho 
=S(\rho). \label{T1}
\end{eqnarray}
Although unified-$(q,s)$ entropy is singular for $q=1$ or $s=0$, we
can consider them to be von Neumann entropy or R\'enyi-$q$ entropy,
respectively. For this reason, we let $S_{1,s}(\rho)\equiv S(\rho)$
and $S_{q,0}(\rho)\equiv R_{q}(\rho)$ for any quantum state $\rho$.

For a bipartite pure state $\ket{\psi}_{AB}$ and each $q,~s \geq 0$,
unified-$(q,s)$ entanglement is
\begin{equation}
E_{q,s}\left(\ket{\psi}_{AB} \right):=S_{q,s}(\rho_A),
\label{TEpure}
\end{equation}
where $\rho_A=\T _{B} \ket{\psi}_{AB}\bra{\psi}$ is the reduced
density matrix for subsystem $A$. For a mixed state $\rho_{AB}$, we
define its  unified-$(q,s)$ entanglement via the convex-roof extension,
\begin{equation}
E_{q,s}\left(\rho_{AB} \right):=\min \sum_i p_i E_{q,s}(\ket{\psi_i}_{AB}),
\label{TEmixed}
\end{equation}
where the minimum is taken over all possible pure state
decompositions of $\rho_{AB}=\sum_{i}p_i
\ket{\psi_i}_{AB}\bra{\psi_i}$.

Because unified-$(q,s)$ entropy converges to R\'enyi and Tsallis
entropies when $s$ tends to 0 and 1 respectively,
\begin{eqnarray}
\lim_{s\rightarrow 0}E_{q,s}\left(\rho_{AB} \right)={\mathcal R}_{q}\left(\rho_{AB} \right),
\label{unirenyi}
\end{eqnarray}
where ${\mathcal R}_{q}\left(\rho_{AB} \right)$ is the R\'enyi-$q$
entanglement of $\rho_{AB}$ ~\cite{ks2}, and
\begin{eqnarray}
\lim_{s\rightarrow 1}E_{q,s}\left(\rho_{AB} \right)={\mathcal T}_{q}\left(\rho_{AB} \right),
\label{uniT}
\end{eqnarray}
where ${\mathcal T}_{q}\left(\rho_{AB} \right)$ is the Tsallis-$q$
entanglement~\cite{KT}. For $q$ tends to 1,
\begin{eqnarray}
\lim_{q\rightarrow1}E_{q,s}\left(\rho_{AB} \right)=E_{\rm f}\left(\rho_{AB} \right),
\end{eqnarray}
where $E_{\rm f}(\rho_{AB})$ is the EoF of $\rho_{AB}$.
Thus unified-$(q,s)$ entanglement is a
two-parameter generalization of EoF.

\subsection{Analytic formula of unified-(q,s) entanglement for two-qubit states}
\label{Subsec: 2formula}

Let us recall concurrence and its
functional relation with EoF in two-qubit systems.
For any bipartite pure state $\ket \psi_{AB}$, its
concurrence,
$\mathcal{C}(\ket \psi_{AB})$ is
\begin{equation}
\mathcal{C}(\ket \psi_{AB})=\sqrt{2(1-\T\rho^2_A)},
\label{pure state concurrence}
\end{equation}
where $\rho_A=\T_B(\ket \psi_{AB}\bra \psi)$~\cite{ww}. For a
mixed state
$\rho_{AB}$, its concurrence is
\begin{equation}
\mathcal{C}(\rho_{AB})=\min \sum_k p_k \mathcal{C}({\ket
{\psi_k}}_{AB}), \label{mixed state concurrence}
\end{equation}
where the minimum is taken over all possible pure state
decompositions, $\rho_{AB}=\sum_kp_k{\ket {\psi_k}}_{AB}\bra
{\psi_k}$.

For a two-qubit pure state $\ket{\psi}_{AB}$ with Schmidt decomposition
\begin{equation}
\ket{\psi}_{AB}=\sqrt{\lambda_0}\ket{00}_{AB}+\sqrt{\lambda_1}\ket{11}_{AB},
\label{schm}
\end{equation}
its reduced density operator of subsystem $A$ is
\begin{equation}
\rho_A=\T_B(\ket \psi_{AB}\bra \psi)=\lambda_0\ket{0}_{A}\bra{0}+\lambda_1\ket{1}_{A}\bra{1}.
\end{equation}
From Eq.~(\ref{pure state concurrence}), we obtain
\begin{equation}
\mathcal{C}(\ket \psi_{AB})=\sqrt{2(1-\T\rho^2_A)}=2\sqrt{\lambda_0\lambda_1},
\label{conlam}
\end{equation}
and, moreover,
\begin{equation}
2\sqrt{\lambda_0\lambda_1}=\left(\T\sqrt{\rho_A}\right)^2-1=
S_{\frac{1}{2}, 2}\left(\rho_A\right)=E_{\frac{1}{2}, 2}\left(\ket{\psi}_{AB}\right),
\label{lamuni}
\end{equation}
where $E_{\frac{1}{2}, 2}\left(\ket{\psi}_{AB}\right)$ is the
unified-$(1/2,2)$ entanglement of $\ket{\psi}_{AB}$. In other words,
unified-$(q,s)$ entanglement of a two-qubit pure state
$\ket{\psi}_{AB}$ coincides with its concurrence for $q=1/2$ and
$s=2$. As both concurrence and unified-$(q,s)$ entanglement of
bipartite mixed states are defined via the convex-roof extension, we
note that unified-$(q,s)$ entanglement of a two-qubit mixed state
reduces to its concurrence when $q=1/2$ and $s=2$;
\begin{equation}
\mathcal{C}(\rho_{AB})=E_{\frac{1}{2}, 2}\left(\rho_{AB}\right),
\label{conuni}
\end{equation}
for a two-qubit state $\rho_{AB}$.

Concurrence has an analytic
formula in two-qubit systems~\cite{ww}. For a two-qubit state $\rho_{AB}$,
\begin{equation}
\mathcal{C}(\rho_{AB})=\max\{0, \lambda_1-\lambda_2-\lambda_3-\lambda_4\},
\label{C_formula}
\end{equation}
where $\lambda_i$'s are the eigenvalues, in decreasing order, of
$\sqrt{\sqrt{\rho_{AB}}\tilde{\rho}_{AB}\sqrt{\rho_{AB}}}$ and
$\tilde{\rho}_{AB}=\sigma_y \otimes\sigma_y
\rho^*_{AB}\sigma_y\otimes\sigma_y$ with the Pauli operator
$\sigma_y$. Furthermore, the relation between concurrence and EoF of
a two-qubit mixed state $\rho_{AB}$ (or a pure state
$\ket{\psi}_{AB} \in \mathbb{C}^2 \otimes \mathbb{C}^{d}$,
$d\geq2$) is given as a monotonically increasing, convex
function such that
\begin{equation}
 E_{\rm f} (\rho_{AB}) = {\mathcal E}(\mathcal{C}\left(\rho_{AB}\right)),
\end{equation}
where
\begin{equation}
{\mathcal E}(x) = H\left(\frac{1-\sqrt{1-x^2}}{2}\right),
\hspace{0.5cm}\mbox{for } 0 \le x \le 1,
\label{eps}
\end{equation}
with the binary entropy function $H(t) = -[t\log t + (1-t)\log
(1-t)]$~\cite{ww}.
The functional relation between concurrence
and EoF as well as the analytic formula of concurrence for two-qubit
states provide an analytic formula
of EoF in two-qubit systems.

Now let us consider the functional relation between unified-$(q,s)$
entanglement and concurrence
for two-qubit states. For any $2\otimes d$ pure
state
$\ket{\psi}_{AB}$ with its Schmidt decomposition
$\ket{\psi}_{AB}=\sqrt{\lambda}\ket{0
0}_{AB}+\sqrt{1-\lambda}\ket{11}_{AB}$, its unified-$(q,s)$
entanglement is
\begin{eqnarray}
E_{q,s}\left(\ket{\psi}_{AB} \right)=S_{q,s}(\rho_A)
=\frac{1}{(1-q)s}\left[\left(\lambda^{q}+\left(1-\lambda\right)^{q}\right)^s-1 \right].
\label{TE2pure}
\end{eqnarray}
Because the concurrence of $\ket{\psi}_{AB}$ is
\begin{eqnarray}
\mathcal{C}(\ket \psi_{AB})=\sqrt{2(1-\T\rho^2_A)}
=2\sqrt{\lambda\left(1-\lambda\right)},
\end{eqnarray}
we have
\begin{equation}
E_{q,s}\left(\ket{\psi}_{AB}
\right)=f_{q,s}\left(\mathcal{C}(\ket \psi_{AB}) \right),
\label{relationpure}
\end{equation}
where $f_{q,s}(x)$ is a differential function
\begin{eqnarray}
f_{q,s}(x):=\frac{\left(\left(1+\sqrt{1-x^2}\right)^{q}
+\left(1-\sqrt{1-x^2}\right)^{q}\right)^s-2^{qs}}{(1-q)s2^{qs}}
\label{f}
\end{eqnarray}
on $0 \leq x \leq 1$. Thus for any $2\otimes d$ pure
state $\ket{\psi}_{AB}$, we have a functional relation between its
concurrence and unified-$(q,s)$ entanglement
for each $q$ and $s$.

We note that $f_{1/2,2}(x)=x$ is the identity function, which also reveals the coincidence of concurrence and
unified-$(1/2,2)$ entanglement in Eq.~(\ref{conuni}).
Furthermore, $f_{q,s}(x)$ converges to ${\mathcal E}(x)$
in Eq.~(\ref{eps}) as $q$ tends to 1, and
it reduces to functions that relate
concurrence with R\'enyi-$q$ entanglement and Tsallis-$q$
entanglement as $s$ tends to $0$ and $1$ respectively~\cite{ks2,
KT}. For two-qubit mixed states, we have the following theorem.

\begin{Thm}
For $q\geq1$, $0 \leq s \leq1$, $qs\leq 3$ and any two-qubit state $\rho_{AB}$,
\begin{equation}
E_{q,s}\left(\rho_{AB}\right)=f_{q,s}\left(\mathcal{C}(\rho_{AB}) \right).
\label{eucmix}
\end{equation}
\label{Thm: 2formula}
\end{Thm}

We note that Theorem~\ref{Thm: 2formula} together with
the analytic formula of two-qubit concurrence in Eq.~(\ref{C_formula})
provide us with an analytic formula of
unified-$(q,s)$ entanglement in two-qubit systems.
Before we prove Theorem~\ref{Thm: 2formula}, we have the following lemma.

\begin{Lem}
For $q\geq1$, $0 \leq s \leq1$ and $qs\leq 3$, $f_{q,s}(x)$
is a monotonically-increasing convex function on $0\leq x \leq 1$.
\label{Lem: fmonocon}
\end{Lem}

\begin{proof}
Because $f_{q,s}(x)$ is a differentiable function on $0\leq x\leq1$,
its monotonicity and convexity follow from nonnegativity of its
first and second derivatives. Furthermore, for $q>1$, the
monotonicity and convexity of $f_{q,s}$ follows from those of a
function,
\begin{equation}
g_{q,s}(x):=-\Big[\left(1+\sqrt{1-x^2}\right)^q+\left(1-\sqrt{1-x^2}\right)^q\Big]^s.
\label{gqs}
\end{equation}

For
\begin{equation}
\Theta=1+\sqrt{1-x^2},~ \Xi =1-\sqrt{1-x^2},
\label{thetaxi}
\end{equation}
the first derivative of $g_{q,s}(x)$ is
\begin{equation}
\frac{{\rm d}g_{q,s}(x)}{{\rm d}x}=\frac{qsx}{\sqrt{1-x^2}}
\left( \Theta^q+\Xi^q\right)^{s-1}\left(\Theta^{q-1}-\Xi^{q-1}\right),
\label{1deri}
\end{equation}
which is always nonnegative on $0\leq x\leq1$ for $q\geq1$. For the
second derivative of $g_{q,s}(x)$, we have
\begin{eqnarray}
\frac{{\rm d}^2g_{q,s}(x)}{{\rm d}x^2}
=&\Lambda \frac{\left(\Theta^q+\Xi^q\right)\left(\Theta^{q-1}-\Xi^{q-1}\right)}{\sqrt{1-x^2}}\nonumber\\
&+\Lambda q(1-s)x^2 \left(\Theta^{q-1}-\Xi^{q-1}\right)^2 \nonumber\\
&-\Lambda(q-1)x^2\left(\Theta^q+\Xi^q\right)\left(\Theta^{q-2}+\Xi^{q-2}\right)
\label{2deri}
\end{eqnarray}
with $\Lambda=qs\left(\Theta^q+\Xi^q\right)^{s-2}/(1-x^2)$.

Because $q(1-s)\geq q-3$ for $qs\leq 3$, we have
\begin{eqnarray}
\frac{{\rm d}^2g_{q,s}(x)}{{\rm d}x^2}
&\geq&\Lambda \left(\Theta^q+\Xi^q\right)
\left[\frac{\left(\Theta^{q-1}-\Xi^{q-1}\right)}{\sqrt{1-x^2}}
-2x^2\left(\Theta^{q-2}+\Xi^{q-2}\right)\right]\nonumber\\
&&+\Lambda (q-3)x^2\left[ \left(\Theta^{q-1}-\Xi^{q-1}\right)^2
-\left(\Theta^q+\Xi^q\right)\left(\Theta^{q-2}+\Xi^{q-2}\right)\right]\nonumber\\
&=&\Lambda\left(\Theta^q+\Xi^q\right)\left[\frac{\left(\Theta^{q-1}-\Xi^{q-1}\right)}{\sqrt{1-x^2}}
-2x^2\left(\Theta^{q-2}+\Xi^{q-2}\right)\right]\nonumber\\
&&-4\Lambda(q-3)x^{2q-2},
\label{2deri2}
\end{eqnarray}
where the equality is given by
\begin{eqnarray}
\left(\Theta^{q-1}-\Xi^{q-1}\right)^2-\left(\Theta^q+\Xi^q\right)\left(\Theta^{q-2}+\Xi^{q-2}\right)=-4x^{2q-4},
\label{squarsimple}
\end{eqnarray}
which is obtained from
\begin{equation}
\Theta+\Xi=2,~\Theta\Xi=x^2.
\end{equation}
From the equality
\begin{equation}
\frac{\Theta^{q-1}-\Xi^{q-1}}{\sqrt{1-x^2}}
=2\left(\Theta^{q-2}+\Xi^{q-2}\right)
+\frac{x^2\left(\Theta^{q-3}-\Xi^{q-3}\right)}{\sqrt{1-x^2}},
\label{ThetaXisqu}
\end{equation}
Eq.~(\ref{2deri2}) becomes
\begin{eqnarray}
\frac{{\rm d}^2g_{q,s}(x)}{{\rm d}x^2}
&\geq&\Lambda\left(\Theta^q+\Xi^q\right)\left[
2(1-x^2)\left(\Theta^{q-2}+\Xi^{q-2}\right)+\frac{x^2\left(\Theta^{q-3}-\Xi^{q-3}\right)}{\sqrt{1-x^2}}
\right]\nonumber\\
&&-4\Lambda(q-3)x^{2q-2}.
\label{2deri3}
\end{eqnarray}

Let us consider the binomial series for $\Theta^{q-1}$ and
$\Xi^{q-1}$,
\begin{eqnarray}
\Theta^{q-1}&=&{\left( 1+ \sqrt{1-x^2}\right)}^{q-1}\nonumber\\
&=&1+(q-1)\sqrt{1-x^2}+\frac{(q-1)(q-2)}{2!}
{\left(\sqrt{1-x^2}\right)}^2+R_1
\label{binomial1}
\end{eqnarray}
and
\begin{eqnarray}
\Xi^{q-1}&=&{\left(1-\sqrt{1-x^2}\right)}^{q-1}\nonumber\\
&=&1-(q-1)\sqrt{1-x^2}+\frac{(q-1)(q-2)}{2!}
{\left(\sqrt{1-x^2}\right)}^2+R_2,
\label{binomial2}
\end{eqnarray}
with remainder terms
\begin{eqnarray}
R_1&=&\sum_{k=3}^{\infty}\frac{(q-1)\cdots(q-k)}{k!}{\left(\sqrt{1-x^2}\right)}^k,\nonumber\\
R_2&=&\sum_{k=3}^{\infty}\frac{(q-1)\cdots(q-k)}{k!}{\left(-\left(\sqrt{1-x^2}\right)\right)}^k.
\label{R1R2}
\end{eqnarray}
Thus we have
\begin{eqnarray}
\Theta^{q-1}-\Xi^{q-1}=2(q-1)\sqrt{1-x^2}+R_1-R_2\geq 2(q-1)\sqrt{1-x^2},\nonumber\\
\Theta^{q-1}+\Xi^{q-1}=2+R_1+R_2 \geq 2,
\label{lower}
\end{eqnarray}
for non-negative constants $R_1-R_2$ and $R_1+R_2$.

From inequality (\ref{2deri3}) together with (\ref{lower}), we have
\begin{eqnarray}
\frac{{\rm d}^2g_{q,s}(x)}{{\rm d}x^2}
\geq&4\Lambda \left[2(1-x^2)+(q-3)(x^2-x^{2q-2})\right],
\label{2deri4}
\end{eqnarray}
where the right-hand side of the inequality is always nonnegative
for $0\leq x\leq 1$ and $q\geq1$. Thus $f_{q,s}(x)$ is monotonically
increasing and convex on $0\leq x\leq 1$ for $q \geq 1$, $0 \leq s
\leq1$ and $qs\leq 3$.
\end{proof}

We note that the monotonicity and convexity of $f_{q,s}(x)$ for
$q\geq1$ and $qs\leq1$ are strict in the sense that the first and
second derivatives of $f_{q,s}(x)$ are strictly positive for
$0<x<1$. Now we prove Theorem~\ref{Thm: 2formula}, which relates
concurrence with unified-$(q,s)$ entanglement for two-qubit mixed
states.

\begin{proof}[Proof of Theorem~\ref{Thm: 2formula}]
For a two-qubit mixed state $\rho_{AB}$ and its concurrence
$\mathcal{C}(\rho_{AB})$, there exists an optimal decomposition of
$\rho_{AB}$, in which every pure-state concurrence has the same
value~\cite{ww}; there exists a pure-state decomposition
$\rho_{AB}=\sum_{i}p_i \ket{\phi_i}_{AB}\bra{\phi_i}$ such that
\begin{equation}
\mathcal{C}(\rho_{AB})=\sum_i p_i \mathcal{C}({\ket {\phi_i}}_{AB}),
\label{Copt}
\end{equation}
and
\begin{equation}
\mathcal{C}(\ket {\phi_i}_{AB})=\mathcal{C}(\rho_{AB}),
\label{Cphii}
\end{equation}
for each $i$.
Thus we have
\begin{eqnarray}
f_{q,s}\left(\mathcal{C}(\rho_{AB}) \right)&=&f_{q,s}\left(\sum_i
p_i\mathcal{C}(\ket{\phi_i}_{AB})\right)\nonumber\\
&=&\sum_i
p_if_{q,s}\left(\mathcal{C}(\ket{\phi_i}_{AB})\right)\nonumber\\
&=&\sum_i
p_iE_{q,s}(\ket{\phi_i}_{AB})\nonumber\\
&\geq&E_{q,s}\left(\rho_{AB}\right).
\label{fmixin1}
\end{eqnarray}

Conversely, the existence of the optimal decomposition of
$\rho_{AB}=\sum_j q_j\ket{\mu_j}_{AB}\bra{\mu_j}$ for unified-$(q,s)$
entanglement leads us to
\begin{eqnarray}
E_{q,s}\left(\rho_{AB}\right)&=&\sum_j q_jE_{q,s}\left(\ket{\mu_j}_{AB}\right)\nonumber\\
&=&\sum_j q_jf_{q,s}\left({\mathcal C}(\ket{\mu_j}_{AB})\right)\nonumber\\
&\geq&f_{q,s}\left(\sum_j q_j{\mathcal C}(\ket{\mu_j}_{AB})\right)\nonumber\\
&\geq&f_{q,s}\left({\mathcal C}(\rho_{AB})\right),\nonumber\\
\label{fmixin2}
\end{eqnarray}
where the first and second inequalities are due to the convexity and
monotonicity of $f_{q,s}(x)$ in Lemma~\ref{Lem: fmonocon}. From
Inequalities (\ref{fmixin1}) and (\ref{fmixin2}), we have
\begin{equation}
E_{q,s}\left(\rho_{AB} \right)=f_{q,s}\left(\mathcal{C}(\rho_{AB}) \right)
\label{relationmixed}
\end{equation}
for $q\geq1$, $0 \leq s \leq1$, $qs\leq3$ and any two-qubit mixed
state $\rho_{AB}$
\end{proof}

Due to the continuity of $f_{q,s}(x)$ with respect to $q$ and $s$,
we can always assure this functional relation between
unified-$(q,s)$ entanglement and concurrence in two-qubit systems
for $q$ slightly less than 1 or $qs$ slightly larger than 3.


\section{Multi-qubit monogamy of entanglement in terms of unified-$(q,s)$ Entanglement}
\label{Sec: monopoly}

The monogamous property of a multi-qubit pure state
$\ket{\psi}_{A_1A_2\cdots A_n}$ has been shown to be
\begin{equation}
\mathcal{C}_{A_1 (A_2 \cdots A_n)}^2  \geq  \mathcal{C}_{A_1 A_2}^2
+\cdots+\mathcal{C}_{A_1 A_n}^2, \label{nCmono}
\end{equation}
where $\mathcal{C}_{A_1 (A_2 \cdots
A_n)}=\mathcal{C}(\ket{\psi}_{A_1(A_2\cdots A_n)})$ is the
concurrence of $\ket{\psi}_{A_1A_2\cdots A_n}$ with respect to the
bipartite cut between $A_1$ and the others, and
$\mathcal{C}_{A_1A_i}=\mathcal{C}(\rho_{A_1A_i})$ is the concurrence
of the reduced density matrix $\rho_{A_1A_i}$ for $i=2,\ldots,
n$~\cite{ckw,ov}.
Here, we show that this monogamy of
multi-qubit entanglement can also be characterized in terms of
unified-$(q,s)$ entanglement. Before we prove multi-qubit monogamy relation of
unified-$(q,s)$ entanglement, we provide an
important property of the function $f_{q,s}(x)$.

\begin{Lem}
For $q\geq2$, $0\leq s \leq1$ and $qs\leq3$,
\begin{equation}
h_{q,s}(x,y):=f_{q,s}\left(\sqrt{x^2+y^2}\right)-f_{q,s}(x)-f_{q,s}(y)\geq 0,
\label{eq: fposi}
\end{equation}
on the domain ${\mathcal D}=\{ (x,y)| 0\leq x, y, x^2+y^2 \leq1\}$.
\label{fposi}
\end{Lem}

\begin{proof}
In fact, Inequality (\ref{eq: fposi}) was already shown
for $s=0$~\cite{ks2} and $s=1$~\cite{KT}. Thus it is enough to
consider the case that $0 < s <1$.

As $h_{q,s}(x,y)$ is differentiable on domain $\mathcal D$, its maximum or
minimum values arise
only at the critical points or on the boundary of $\mathcal D$. By
taking the first-order partial derivatives, we have
the gradient of $h_{q,s}(x,y)$ as
\begin{equation}
\nabla h_{q,s}(x, y)=\left(\frac{\partial
h_{q,s}(x,y)}{\partial x}, \frac{\partial
h_{q,s}(x,y)}{\partial y}\right)
\label{grad}
\end{equation}
where
\begin{eqnarray}
\frac{\partial h_{q,s}(x,y)}{\partial x}&=
&\Gamma \frac{qsx\left[{\left(1+\sqrt{1-x^2}\right)}^{q}+
{\left(1-\sqrt{1-x^2}\right)}^{q}\right]^{s-1}}{\sqrt{1-x^2}}
\nonumber\\
&&\times
\left[{\left(1+\sqrt{1-x^2}\right)}^{q-1}-{\left(1-\sqrt{1-x^2}\right)}^{q-1}\right]\nonumber\\
&-&\Gamma \frac{qsx\left[{\left(1+\sqrt{1-x^2-y^2}\right)}^{q}+
{\left(1-\sqrt{1-x^2-y^2}\right)}^{q}\right]^{s-1}}{\sqrt{1-x^2-y^2}}\nonumber\\
&&\times\left[{\left(1+\sqrt{1-x^2-y^2}\right)}^{q-1}-
{\left(1-\sqrt{1-x^2-y^2}\right)}^{q-1}\right],\nonumber\\
\frac{\partial h_{q,s}(x,y)}{\partial y}&=
&\Gamma \frac{qsy\left[{\left(1+\sqrt{1-y^2}\right)}^{q}+
{\left(1-\sqrt{1-y^2}\right)}^{q}\right]^{s-1}}{\sqrt{1-y^2}}
\nonumber\\
&&\times
\left[{\left(1+\sqrt{1-y^2}\right)}^{q-1}-{\left(1-\sqrt{1-y^2}\right)}^{q-1}\right]\nonumber\\
&-&\Gamma \frac{qsy\left[{\left(1+\sqrt{1-x^2-y^2}\right)}^{q}+
{\left(1-\sqrt{1-x^2-y^2}\right)}^{q}\right]^{s-1}}{\sqrt{1-x^2-y^2}}\nonumber\\
&&\times\left[{\left(1+\sqrt{1-x^2-y^2}\right)}^{q-1}-{\left(1-\sqrt{1-x^2-y^2}\right)}^{q-1}\right]
\label{2pderi}
\end{eqnarray}
with $\Gamma=1/\left[(1-q)s2^{sq}\right]$.

Now, let us suppose there exists $(x_0, y_0)$ in the interior of the domain
${\mathcal D}^{\circ}=\{ (x,y)| 0< x,~y,~x^2+y^2 <1\}$
such that $\nabla h_{q,s}(x_0,
y_0)=(0,0)$. From Eq.~(\ref{2pderi}), it is straightforward to verify
that $\nabla h_{q,s}(x_0, y_0)=(0,0)$ implies
\begin{equation}
n_{q,s}(x_0)=n_{q,s}(y_0),
\label{x0y0}
\end{equation}
for a differentiable function
\begin{eqnarray}
n_{q,s}(x)&:=&\frac{qs}{\sqrt{1-x^2}}
\left[{\left(1+\sqrt{1-x^2}\right)}^{q}+{\left(1-\sqrt{1-x^2}\right)}^{q}\right]^{s-1}\nonumber\\
&&\times\left[{\left(1+\sqrt{1-x^2}\right)}^{q-1}-{\left(1-\sqrt{1-x^2}\right)}^{q-1}\right],
\label{n_q}
\end{eqnarray}
defined on $0<x<1$. Here we first show that $n_{q,s}(x)$ is a
strictly increasing function for $0 < x < 1$, and thus
Eq.~(\ref{x0y0}) implies $x_0=y_0$.

Let us consider the first derivative of $n_{q,s}(x)$. Because
$xn_{q,s}(x)={{\rm d}g_{q,s}(x)}/{{\rm d}x}$, where ${{\rm
d}g_{q,s}(x)}/{{\rm d}x}$ is in Eq.~(\ref{1deri}), we have
\begin{eqnarray}
\frac{{\rm d}n_{q,s}(x)}{{\rm d}x}=\frac{1}{x}\left(\frac{{\rm d^2}g_{q,s}(x)}{{\rm d}x^2}-n_{q,s}(x)\right),
\label{nderi1}
\end{eqnarray}
for non-zero $x$. To show $n_{q,s}(x)$ is strictly increasing
function, it is thus enough to show that ${{\rm
d^2}g_{q,s}(x)}/{{\rm d}x^2}-n_{q,s}(x)>0$ for $0< x< 1$. By using
$\Theta$ and $\Xi$ (\ref{thetaxi}), we have
\begin{eqnarray}
\frac{{\rm d^2}g_{q,s}(x)}{{\rm d}x^2}-n_{q,s}(x)&=&
\Omega\frac{x^2\left(\Theta^q+\Xi^q\right)\left(\Theta^{q-1}-\Xi^{q-1}\right)}{\sqrt{1-x^2}}\nonumber\\
&&+\Omega q(1-s)x^2\left(\Theta^{q-1}-\Xi^{q-1}\right)^2\nonumber\\
&&-\Omega(q-1)x^2\left(\Theta^q+\Xi^q\right)\left(\Theta^{q-2}+\Xi^{q-2}\right),
\label{g2dn}
\end{eqnarray}
with
$\Omega=qs\left(\Theta^q+\Xi^q \right)^{s-2}/\left( 1-x^2\right)$.

Because $q(1-s)\geq q-3$ for $qs\leq3$,
\begin{eqnarray}
\frac{{\rm d^2}g_{q,s}(x)}{{\rm d}x^2}-n_{q,s}(x)&\geq&
\Omega x^2\left(\Theta^q+\Xi^q\right)\left[\frac{\left(\Theta^{q-1}-\Xi^{q-1}\right)}{\sqrt{1-x^2}}
-2\left(\Theta^{q-2}+\Xi^{q-2}\right)\right]\nonumber\\
&&+\left[\left(\Theta^{q-1}-\Xi^{q-1}\right)^2-
\left(\Theta^q+\Xi^q\right)\left(\Theta^{q-2}+\Xi^{q-2}\right)\right]\nonumber\\
&&\times\Omega(q-3)x^2.
\label{g2dn2}
\end{eqnarray}
Furthermore, from the relation $\Theta-\Xi=2\sqrt{1-x^2}$, we have
\begin{eqnarray}
\frac{\Theta^{q-1}-\Xi^{q-1}}{\sqrt{1-x^2}}
-2\left(\Theta^{q-2}+\Xi^{q-2}\right)=&\frac{x^2\left(\Theta^{q-3}-\Xi^{q-3}\right)}{\sqrt{1-x^2}}.
\label{g2dn2rel}
\end{eqnarray}
Together with Eq.~(\ref{squarsimple}) and Inequalities (\ref{lower}), we have
\begin{eqnarray}
\frac{{\rm d^2}g_{q,s}(x)}{{\rm d}x^2}-n_{q,s}(x) \geq&
4\Omega(q-3)(x^4-x^{2q-2}). \label{g2dn3}
\end{eqnarray}

Here we note that the right-hand side of the inequality~(\ref{g2dn3})
is strictly positive for $0<x<1$ when $q\neq 3$; therefore,
$n_{q,s}(x)$ is a strictly increasing function for $q\neq 3$ and $qs\leq 3$.
For the case that $q=3$, we have
\begin{equation}
n_{3,s}(x)=12s\left(8-6x^2\right)^{s-1},
\label{n3s}
\end{equation}
which is also a strictly increasing function for $0<s<1$. In other
words, $n_{q,s}(x)$ is a strictly increasing function for $q\geq2$,
$0<s<1$ and $qs\leq3$, therefore Eq.~(\ref{x0y0}) implies $x_0=y_0$.
However, from Eq.~(\ref{2pderi}), ${\partial
h_{q,s}(x_0,x_0)}/{\partial x}=0$ also implies that
$n_{q,s}(x_0)=n_{q,s}(\sqrt{2}x_0)$ for some $x_0 \in (0,1)$, which
contradicts the strict monotonicity of $n_{q,s}(x)$; $n_{q,s}(x)$
has non-vanishing gradient in ${\mathcal D}^{\circ}$ for $q\geq1$
and $qs\leq 3$.

Now let us consider the function value of $h_{q,s}(x, y)$ on the
boundary of the domain $\partial\mathcal D=\{(x,y)|
x=0~or~y=0~or~x^2+y^2=1\}$. If either $x$ or $y$ is 0, then it is
clear that $h_{q,s}(x, y)=0$. For the case that $x^2+y^2=1$,
$h_{q,s}(x, y)$ is reduced to a single-variable function,
\begin{eqnarray}
l_{q,s}(x)
:=&\frac{1}{(q-1)s2^{qs}}\left(\left(1+\sqrt{1-x^2}\right)^{q}+\left(1-\sqrt{1-x^2}\right)^{q}\right)^{s}\nonumber\\
&+\frac{1}{(q-1)s2^{qs}}\left[\left(\left(1+x\right)^{q}+\left(1-x\right)^{q}\right)^{s}-2^s-2^{qs}\right].
\label{lqs}
\end{eqnarray}
In other words, the nonnegativity of $h_{q,s}(x,y)$ for $q\geq2$ and
$qs\leq 3$ follows from that of the differentiable function
\begin{eqnarray}
m_{q,s}(x):=&\left(\left(1+\sqrt{1-x^2}\right)^{q}+\left(1-\sqrt{1-x^2}\right)^{q}\right)^{s}\nonumber\\
&+\left(\left(1+x\right)^{q}+\left(1-x\right)^{q}\right)^{s}-2^s-2^{qs}.
\label{mqs}
\end{eqnarray}

From the derivative of $m_{q,s}(x)$,
\begin{eqnarray}
\frac{{\rm d}m_{q,s}(x)}{{\rm d}x}&=&sq\left[\left(1+x\right)^{q}+\left(1-x\right)^{q}\right]^{s-1}
\left[\left(1+x\right)^{q-1}-\left(1-x\right)^{q-1}\right]\nonumber\\
&&-\frac{sqx}{\sqrt{1-x^2}}\left[\left(1+\sqrt{1-x^2}\right)^{q}+
\left(1-\sqrt{1-x^2}\right)^{q}\right]^{s-1}\nonumber\\
&&\times\left[\left(1+\sqrt{1-x^2}\right)^{q-1}-\left(1-\sqrt{1-x^2}\right)^{q-1}\right],
\label{mqsderi}
\end{eqnarray}
we note that $x=1/\sqrt{2}$ is the only critical point of
$m_{q,s}(x)=0$ on $0 < x  <1$. Because $m_{q,s}(0)=m_{q,s}(1)=0$ and
$m_{q,s}(x)$ has only one critical point on $0 < x  <1$,
$m_{q,s}(x)$ is either nonnegative or nonpositive through the whole
range of $0\leq x \leq1$.

To show $m_{q,s}(x)$ is nonnegative on $0\leq x \leq1$, we show its
nonnegativity for $x$ near $1$. Let us consider the derivative of
$m_{q,s}(x)$ as $x$ approaches $1$. From Eq.~(\ref{mqsderi}) and
Inequalities (\ref{lower}), we note that
\begin{eqnarray}
\frac{{\rm d}m_{q,s}(x)}{{\rm d}x}&\leq& sq\left[\left(1+x\right)^{q}+\left(1-x\right)^{q}\right]^{s-1}
\left[\left(1+x\right)^{q-1}-\left(1-x\right)^{q-1}\right]\nonumber\\
&&-sqx2(q-1)\times\left[\left(1+\sqrt{1-x^2}\right)^{q}+\left(1-\sqrt{1-x^2}\right)^{q}\right]^{s-1}
\label{mqsderi2}
\end{eqnarray}
for $0\leq x \leq1$; therefore
\begin{eqnarray}
\lim_{x\rightarrow 1}\frac{{\rm d}m_{q,s}(x)}{{\rm d}x} \leq sq\left[2^{qs-1}-(q-1)2^s\right].
\label{derilim}
\end{eqnarray}

For $qs \leq 3$ and $0 \leq s\leq1$, $2^{qs-1}$ in the right-hand
side of Inequality (\ref{derilim}) is bounded above by $4$, whereas
$(q-1)2^s\geq q-1$. Thus, the right-hand side of Inequality
(\ref{derilim}) is always negative for $q >5$. In other words,
$m_{q,s}(x)$ is a decreasing function as $x$ approaches to $1$ with
$m_{q,s}(1)=0$, and thus $m_{q,s}(x)$ is a nonnegative function for
$q>5$.

For $q\leq5$, we consider the function value of a two-variable function
$b(q,s)=2^{qs-1}-(q-1)2^s$
on the compact domain ${\mathcal D}_2=\{(q,s)| 2\leq q \leq 5, 0\leq s\leq1, qs \leq 3\}$.
The first-order partial derivatives of $b(q,s)$ are
\begin{eqnarray}
\frac{\partial b(q,s)}{\partial q}=2^{qs-1}s\log2-2^s,
\label{bderiq}
\end{eqnarray}
and
\begin{eqnarray}
\frac{\partial b(q,s)}{\partial s}=2^{qs-1}q\log2-(q-1)2^s\log2.
\label{bderis}
\end{eqnarray}

If we assume $b(q,s)$ has a critical point at $(q_0, s_0)$ in the interior of the
domain ${\mathcal D}_2^{\circ}=\{(q,s)| 2 < q < 5, 0< s<1, qs < 3\}$,
Eq.~(\ref{bderiq}) implies
\begin{equation}
\frac{\partial b(q_0,s_0)}{\partial q}=2^{q_0s_0-1}s_0\log2-2^{s_0}=0.
\label{bderiq0}
\end{equation}
Furthermore, from Eq.~(\ref{bderis}) together with Eq.~(\ref{bderiq0}), we have
\begin{eqnarray}
\frac{\partial b(q_0,s_0)}{\partial s}&=&2^{q_0s_0-1}q_0\log2-(q_0-1)2^{s_0}\log2\nonumber\\
&=&\frac{2^{s_0}}{s_0}\left[q_0-(q_0-1)s_0\log2\right].
\label{bderis0}
\end{eqnarray}

However $s_0\log2$ is strictly less than $1$, therefore
Eq.~(\ref{bderis0}) is always nonzero in the interior of the domain.
In other words, for any  $(q_0, s_0)$ in the interior of the domain,
${\partial b(q_0,s_0)}/{\partial s}$ is always nonzero conditioned
${\partial b(q_0,s_0)}/{\partial q}=0$.  Thus $b(q,s)$ has no
vanishing gradient in the interior of the domain. Furthermore, it is
also direct to verify that $b(q,s)$ is non-positive on the boundary
of the domain, and thus $b(q,s)$ is non-positive for $2\leq q \leq
5$, $0<s<1$ and $qs \leq 3$.

Thus $m_{q,s}(x)$ is nonnegative on $0 \leq x \leq1$ for $0\leq
s\leq 1$, $q\geq 2$ and $qs \leq3$, and this implies nonnegativity
of $h_{q,s}(x,y)$ for the same domain of $q$ and $s$.
\end{proof}

The following theorem yields
a multi-qubit monogamy inequality in terms of unified-$(q,s)$
entanglement.

\begin{Thm}
For $q\geq2$, $0\leq s \leq1$, $qs\leq3$ and a multi-qubit state
$\rho_{A_1 \cdots A_n}$, we have
\begin{equation}
E_{q,s}\left( \rho_{A_1(A_2 \cdots A_n)}\right)\geq
E_{q,s}(\rho_{A_1 A_2}) +\cdots+E_{q,s}(\rho_{A_1
A_n}) \label{Umono}
\end{equation}
where $E_{q,s}\left( \rho_{A_1(A_2 \cdots A_n)}\right)$ is the
unified-$(q,s)$ entanglement of $\rho_{A_1\left(A_2 \cdots
A_n\right)}$ with respect to the bipartite cut between $A_1$ and
$A_{2}\cdots A_{n}$, and $E_{q,s}(\rho_{A_1 A_i})$ is the
unified-$(q,s)$ entanglement of the reduced state
$\rho_{A_1 A_i}$ for $i=2,\cdots,n$. \label{Thm: mono}
\end{Thm}

\begin{proof}
We first prove the theorem for $n$-qubit pure state
$\ket{\psi}_{A_1\cdots A_n}$. Note Inequality Eq.~(\ref{nCmono}) is
equivalent to
\begin{equation}
\mathcal{C}_{A_1 (A_2 \cdots A_n)} \geq  \sqrt{\mathcal{C}_{A_1
A_2}^2 +\cdots+\mathcal{C}_{A_1 A_n}^2}, \label{nCmonoroot}
\end{equation}
for any $n$-qubit pure state $\ket{\psi}_{A_1 (A_2 \cdots A_n)}$.
Thus, from Lemma~\ref{fposi} together with Eq.~(\ref{nCmonoroot}),
we have
\begin{eqnarray}
E_{q,s}\left(\ket{\psi}_{A_1 (A_2 \cdots A_n)} \right)&=&
f_{q,s}\left(\mathcal{C}_{A_1 (A_2 \cdots A_n)}\right)\nonumber\\
&\geq&
f_{q,s}\left(\sqrt{\mathcal{C}_{A_1 A_2}^2 +\cdots+\mathcal{C}_{A_1 A_n}^2}\right)\nonumber\\
&\geq&f_{q,s}\left(\mathcal{C}_{A_1
A_2}\right)+f_{q,s}\left(\sqrt{\mathcal{C}_{A_1 A_3}^2
+\cdots+\mathcal{C}_{A_1 A_n}^2}\right)\nonumber\\
&&~~~~~~~\vdots\nonumber\\
&\geq& f_{q,s}\left(\mathcal{C}_{A_1 A_2}\right)+\cdots+f_{q,s}\left(\mathcal{C}_{A_1 A_n}\right)\nonumber\\
&=& E_{q,s}\left(\rho_{A_1A_2}\right)+\cdots +E_{q,s}\left(\rho_{A_1A_n}\right), \label{monoineq}
\end{eqnarray}
where the first equality is by the functional relation between the
concurrence and the unified-$(q,s)$ entanglement for $2\otimes d$ pure
states, the first inequality is by the monotonicity of $f_{q,s}(x)$,
the other inequalities are by iterative use of Lemma~\ref{fposi},
and the last equality is by Theorem~\ref{Thm: 2formula}.

For an $n$-qubit mixed state  $\rho_{A_1(A_2\cdots A_n)}$, let
$\rho_{A_1(A_2\cdots A_n)}=\sum_j p_j \ket{\psi_j}_{A_1(A_2\cdots
A_n)}\bra{\psi_j}$ be an optimal decomposition such that
$E_{q,s}\left(\rho_{A_1(A_2\cdots A_n)}\right)=\sum_j p_j
E_{q,s}\left(\ket{\psi_j}_{A_1(A_2\cdots A_n)}\right)$. Because each
$\ket{\psi_j}_{A_1(A_2\cdots A_n)}$ in the decomposition is an
$n$-qubit pure state, we have
\begin{eqnarray}
E_{q,s}\left(\rho_{A_1(A_2\cdots A_n)}\right)&=&\sum_j p_j
E_{q,s}\left(\ket{\psi_j}_{A_1(A_2\cdots
A_n)}\right)\nonumber\\
&\geq&\sum_j p_j\left(E_{q,s}\left(\rho^j_{A_1A_2}\right)+\cdots
+E_{q,s}\left(\rho^j_{A_1A_n}\right) \right)\nonumber\\
&=&\sum_j p_jE_{q,s}\left(\rho^j_{A_1A_2}\right)+\cdots
+\sum_j p_jE_{q,s}\left(\rho^j_{A_1A_n}\right) \nonumber\\
&\geq&E_{q,s}\left(\rho_{A_1A_2}\right)+\cdots +E_{q,s}\left(\rho_{A_1A_n}\right), \label{Tmonomixed}
\end{eqnarray}
where the last inequality is by definition of unified-$(q,s)$
entanglement for each $\rho_{A_1A_i}$.
\end{proof}

Theorem~\ref{Thm: mono} is a direct consequence of Lemma~\ref{fposi}
when there is a functional relation between unified-$(q,s)$
entanglement and concurrence in two-qubit systems. Here we note that
Lemma~\ref{fposi} is also a necessary condition for multi-qubit
monogamy inequality in terms of unified-$(q,s)$ entanglement: for a
three-qubit W-class state~\cite{DVC}
\begin{equation}
\ket{\mathrm{W}}_{ABC}=a\ket{100}_{ABC}+b\ket{001}_{ABC}+c\ket{010}_{ABC}
\label{3W}
\end{equation}
with $|a|^2+|b|^2+|c|^2=1$, it is straightforward to verify that
\begin{eqnarray}
&\mathcal{C}\left(\ket{\mathrm{W}}_{A(BC)}\right)=\sqrt{2|a|^2\left(|b|^2+|c|^2\right)},\nonumber\\
&\mathcal{C}\left(\rho_{AB}\right)=\sqrt{2|a|^2|b|^2},~
\mathcal{C}\left(\rho_{AC}\right)=\sqrt{2|a|^2|c|^2},
\label{wC}
\end{eqnarray}
where $\mathcal{C}\left(\ket{\mathrm{W}}_{A(BC)}\right)$ is the concurrence of
$\ket{\mathrm{W}}_{ABC}$ with respect to the bipartite cut between $A$ and $BC$, and
$\mathcal{C}\left(\rho_{AB}\right)$ and $\mathcal{C}\left(\rho_{AC}\right)$
are the concurrences of the reduced density matrices
$\rho_{AB}=\T_C\ket{\mathrm{W}}_{ABC}\bra{\mathrm{W}}$ and
$\rho_{AC}=\T_B\ket{\mathrm{W}}_{ABC}\bra{\mathrm{W}}$ respectively.
In other words, the CKW inequality (\ref{nCmono})
is saturated by $\ket{\mathrm{W}}_{ABC}$,
\begin{equation}
\mathcal{C}\left(\ket{\mathrm{W}}_{A(BC)}\right)^2
=\mathcal{C}\left(\rho_{AB}\right)^2+\mathcal{C}\left(\rho_{AC}\right)^2.
\label{wsatu}
\end{equation}

Now suppose there is $(x_0,y_0)$ in the domain ${\mathcal D}$ of the
function $h_{q,s}(x,y)$ in Lemma~\ref{fposi} where the
inequality~(\ref{eq: fposi}) does not hold;
\begin{equation}
f_{q,s}\left(\sqrt{x_0^2+y_0^2}\right)-f_{q,s}(x_0)-f_{q,s}(y_0)<0.
\label{eq: fneg}
\end{equation}
In this case, we can always find a W-class state in Eq.~(\ref{3W}) such that
\begin{eqnarray}
x_0=\sqrt{2|a|^2|b|^2}=\mathcal{C}\left(\rho_{AB}\right),~
y_0=\sqrt{2|a|^2|c|^2}=\mathcal{C}\left(\rho_{AC}\right),
\label{abcxy}
\end{eqnarray}
and thus
\begin{eqnarray}
E_{q,s}\left(\ket{\mathrm{W}}_{{A(BC)}}\right)&=
&f_{q,s}\left(\mathcal{C}\left(\ket{\mathrm{W}}_{A(BC)}\right)\right)\nonumber\\
&=&f_{q,s}\left(\sqrt{\mathcal{C}\left(\rho_{AB}\right)^2+\mathcal{C}\left(\rho_{AC}\right)^2}\right)\nonumber\\
&<&f_{q,s}\left(\mathcal{C}\left(\rho_{AB}\right)\right)+
f_{q,s}\left(\mathcal{C}\left(\rho_{AC}\right)\right)\nonumber\\
&=&E_{q,s}\left(\rho_{AB}\right)+E_{q,s}\left(\rho_{AC}\right),
\label{Wcount1}
\end{eqnarray}
which is a violation of the inequality in (\ref{Umono}). Thus,
Lemma~\ref{fposi} is a necessary and sufficient condition for
multi-qubit monogamy inequality in terms of unified-$(q,s)$
entanglement.

Although unified-$(q,s)$ entanglement reduces to concurrence when
$q=1/2$ and $s=2$, this case does not satisfy the condition of
Theorem~\ref{Thm: mono} for multi-qubit monogamy inequality.
However, we note that the CKW inequality (\ref{nCmono})
characterizes the monogamy of multi-qubit entanglement in terms of
squared concurrence rather than concurrence itself. In fact,
Inequality~(\ref{wsatu}) also implies that monogamy inequality of
multi-qubit entanglement fails if we use concurrence rather than its
square; for non-zero $\mathcal{C}\left(\rho_{AB}\right)$ and
$\mathcal{C}\left(\rho_{AC}\right)$ in (\ref{wsatu}), we have
\begin{equation}
\mathcal{C}\left(\ket{\mathrm{W}}_{A(BC)}\right)^2
=\mathcal{C}\left(\rho_{AB}\right)^2+\mathcal{C}\left(\rho_{AC}\right)^2\lneqq
\left(\mathcal{C}\left(\rho_{AB}\right)+\mathcal{C}\left(\rho_{AC}\right)\right)^2,
\label{wsatu2}
\end{equation}
and thus
\begin{equation}
\mathcal{C}\left(\ket{\mathrm{W}}_{A(BC)}\right)\lneqq
\mathcal{C}\left(\rho_{AB}\right)+\mathcal{C}\left(\rho_{AC}\right).
\label{wsatu2}
\end{equation}
Strictly speaking, concurrence does not show monogamy inequality of
two-qubit entanglement whereas its square does in forms of CKW inequality.

For a bipartite pure state $\ket{\psi}_{AB}$, the squared
concurrence is also referred as tangle
\begin{equation}
\tau\left(\ket{\psi_{AB}}\right):=\mathcal{C}\left(\ket{\psi_{AB}}\right)^2=2\left(1-\T \rho_A^2\right),
\label{tpur}
\end{equation}
and it is also extended to mixed states via the convex-roof
extension,
\begin{equation}
\tau\left(\rho_{AB}\right):=\min\sum_i p_i \left(\mathcal{C}(\ket{\psi_i}_{AB})\right)^2
=\min\sum_i p_i \tau(\ket{\psi_i}_{AB}), \label{tmixed}
\end{equation}
among all the pure state ensembles representing $\rho_{AB}$~\cite{ckw}.
Thus tangle is always an upper bound of the squared concurrence for bipartite mixed
state~\cite{var},
\begin{eqnarray}
\tau\left(\rho_{AB}\right)&=&\min\sum_i p_i \mathcal{C}(\ket{\psi_i}_{AB})^2\nonumber\\
&\geq&\left(\min\sum_i p_i \left(\mathcal{C}(\ket{\psi_i}_{AB})\right)\right)^2\nonumber\\
&=&\mathcal{C}(\rho_{AB})^2.
\label{tC}
\end{eqnarray}

In two-qubit systems, however, Eqs.~(\ref{Copt}) and (\ref{Cphii})
imply the existence of an optimal decomposition of $\rho_{AB}$, in
which every pure-state concurrence has the same value, and thus
Inequality~(\ref{tC}) is always saturated in two-qubit systems;
\begin{eqnarray}
\tau\left(\rho_{AB}\right)=\mathcal{C}(\rho_{AB})^2,
\label{tC2}
\end{eqnarray}
for any two-qubit state $\rho_{AB}$. In other words, the CKW inequality~(\ref{nCmono})
can be rephrased as
\begin{equation}
\tau\left(\rho_{A_1(A_2\cdots A_n)}\right) \geq
\tau\left(\rho_{A_1A_2} \right)+\cdots +
\tau\left(\rho_{A_1A_n}\right), \label{nTmonomixed}
\end{equation}
for any $n$-qubit state $\rho_{A_1A_2\cdots A_n}$~\cite{ov}.

Here we note that tangle is in fact a special case of unified-$(q,s)$ entanglement.
For $q=2$, $s=1$ and a bipartite pure state$\ket{\psi}_{AB}$, we have
\begin{equation}
E_{2,1}\left(\ket{\psi}_{AB}\right)=S_{2,1}\left(\rho_A \right)
=1-\T \rho_A^2=\frac{\tau\left(\ket{\psi}_{AB}\right)}{2}.
\label{tuni}
\end{equation}
As both tangle and unified entanglement are extended to mixed states
via the convex-roof extension, tangle can be considered as
unified-$(2,1)$ entanglement up to a constant factor; therefore
Inequality~(\ref{nTmonomixed}) is equivalent to the monogamy
inequality in terms of unified-$(2,1)$ entanglement. In other words,
Inequality (\ref{Umono}) in Theorem~\ref{Thm: mono} reduces to the
CKW inequality when $q=2$, $s=1$.
\begin{figure}
\begin{center}
\includegraphics[width=7cm]{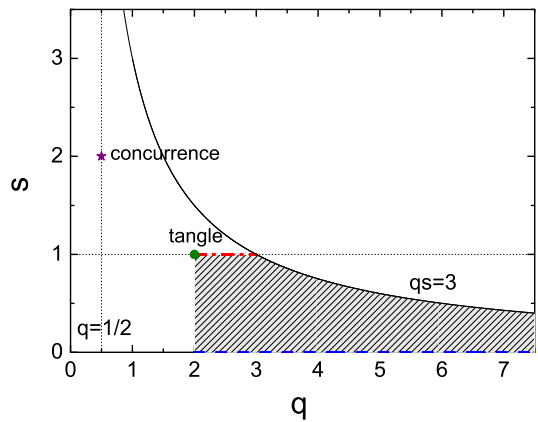}\\
\caption{(Color online) The domain of $q$ and $s$ where multi-qubit monogamy inequality holds
in terms of unified-$(q,s)$ entanglement.
The dashed line indicates the domain for which the multi-qubit monogamy inequality holds for R\'enyi-$q$ entanglement,
and the dashed-dot line is the domain for Tsallis-$q$ entanglement.
The shaded range is for unified-$(q,s)$ entanglement.}
\end{center}
\label{fig1}
\end{figure}

We also note that Inequality (\ref{Umono}) is reduced to the
R\'enyi-$q$ monogamy inequality~\cite{ks2}
\begin{equation}
{\mathcal R}_{q}\left( \rho_{A_1(A_2 \cdots A_n)}\right)\geq
{\mathcal R}_{q}(\rho_{A_1 A_2}) +\cdots+{\mathcal R}_{q}(\rho_{A_1
A_n})
\label{Rmono}
\end{equation}
for $s \rightarrow 0$. For the case that $s \rightarrow 1$,
Inequality (\ref{Umono}) reduces to the Tsallis-$q$ monogamy
inequality~\cite{KT}
\begin{equation}
{\mathcal T}_{q}\left( \rho_{A_1(A_2 \cdots A_n)}\right)\geq
{\mathcal T}_{q}(\rho_{A_1 A_2}) +\cdots+{\mathcal T}_{q}(\rho_{A_1
A_n}). \label{Rmono}
\end{equation}
Thus, Theorem~\ref{Thm: mono} provides an interpolation between
R\'enyi and Tsallis monogamy inequalities as well as the CKW
inequality, which is illustrated in Fig.~1.

We further note that the continuity of unified-$(q,s)$ entropy also
guarantees multi-qubit monogamy inequality in terms of
unified-$(q,s)$ entanglement when $q$ and $s$ are slightly outside
of the proposed domain in Fig.~1.


\section{Conclusion}
\label{Conclusion}

Using unified-$(q,s)$ entropy, we have provided a two-parameter
class of bipartite entanglement measures, namely unified-$(q,s)$
entanglement with an analytical formula in two-qubit systems for
$q\geq 1$, $0\leq s \leq1$ and $qs\leq3$. Based on this unified
formalism of entropies, we have established a broad class of
multi-qubit monogamy inequalities in terms of unified-$(q,s)$
entanglement for $q\geq2$, $0\leq s \leq1$ and $qs\leq3$.

Our new class of monogamy inequalities reduces to every known case
of multi-qubit monogamy inequality such as the CKW inequality,
R\'enyi and Tsallis monogamy inequalities for selective choices of
$q$ and $s$. Our result also provides a necessary and sufficient
condition for a multi-qubit monogamy inequality in terms of
unified-$(q,s)$ entanglement. Furthermore, the explicit relation
between different monogamy inequalities was derived with respect to
a smooth function $f_{q,s}(x)$. Thus, our result provides a useful
methodology to understand the monogamous property of multi-party
entanglement.

\section*{Acknowledgments}
This work was supported by {\it i}CORE, MITACS and USARO. BSC is
supported by a CIFAR Fellowship.


\section*{References}
\begin{thebibliography}{9}

\bibitem{tele}
Bennett~C~H, Brassard~G, Crepeau~C, Jozsa~R, Peres~A and
Wootters~W~K 1993 Phys. Rev. Lett. {\bf 70} 1895

\bibitem{qkd1}
Bennett~C~H and Brassard~G 1984
{\em Quantum Cryptography: Public Key Distribution and Coin Tossing}
in Proceedings of IEEE International Conference on Computers, Systems, and Signal
Processing (IEEE Press, New York, Bangalore, India) pp. 175--179

\bibitem{qkd2}
Bennett~C~H 1992
Phys. Rev. Lett. {\bf 68} 3121

\bibitem{ckw}
Coffman~V, Kundu~J and Wootters~W~K 2000
{\it Phys. Rev.} A {\bf 61} 052306

\bibitem{ov}
Osborne~V and Verstraete F~2006
Phys. Rev. Lett. {\bf 96} 220503

\bibitem{kw}
Koashi M and Winter A 2004 Phys. Rev. A {\bf69} 022309

\bibitem{T04}
Terhal B M 2004
IBM J. Research and Development {\bf48} 71

\bibitem{ww}
Wootters~W~K 1998
Phys. Rev. Lett. {\bf 80} 2245

\bibitem{kds}
Kim~J~S, Das~A and Sanders~B~C 2009
Phys. Rev. A {\bf 79} 012329

\bibitem{gbs}
Gour~G, Bandyopadhay~S and Sanders~B~C 2007
J. Math. Phys. {\bf 48} 012108

\bibitem{bgk}
Buscemi~F, Gour~G and Kim~J~S 2009
Phys. Rev. A {\bf 80} 012324

\bibitem{kpoly}
Kim~J~S 2009
Phys. Rev. A {\bf 80} 022302

\bibitem{ou}
Ou~Y 2007
Phys. Rev. A {\bf 75} 034305

\bibitem{ks}
Kim~J~S and Sanders~B~C 2008
J. Phys. A: Math. and Theor. {\bf 41} 495301

\bibitem{bdsw}
Bennett~C~H, DiVincenzo~D~P, Smolin~J~A and Wootters~W~K~1996
Phys. Rev. A {\bf 54} 3824

\bibitem{renyi}
R\'enyi A 1960
{\em On Measures of Information and Entropy}
in Proceedings of the Fourth Berkeley
Symposium on Mathematics, Statistics and Probability
(Berkeley University Press, Berkeley, CA) pp. 547--561

\bibitem{horo}
Horodecki~R, Horodecki~P and Horodecki M~1996
Phys. Lett. A {\bf 210} 377

\bibitem{tsallis}
Tsallis~C 1988 J. Stat. Phys. {\bf 52} 479

\bibitem{lv}
Landsberg~P~T and Vedral~V 1998 Phys. Lett. A {\bf 247} 211

\bibitem{ue1}
Hu~X and Ye~Z 2006 J. Math. Phys. {\bf 47} 023502

\bibitem{ue2}
Rastegin~A~E 2010 arXiv: 1012.5356

\bibitem{ks2}
Kim~J~S and Sanders~B~C 2010
J. Phys. A: Math. and Theor. {\bf 43} 445305

\bibitem{KT}
Kim~J~S 2010
Phys. Rev. A. {\bf 81} 062328

\bibitem{DVC}
D\"{u}r~W, Vidal~G and Cirac J~I 2000
Phys. Rev. A {\bf 62} 062314

\bibitem{var}
Because $\tau\left(\rho_{AB}\right)$ is the minimum average of squared concurrence,
the difference between $\tau\left(\rho_{AB}\right)$ and ${\mathcal C} \left(\rho_{AB}\right)^2$ is
the {\em variance}, which is always nonnegative.
Furthermore, the equality holds if every
pure-state concurrence in the decomposition has the same value.
\end{thebibliography}
\end{document}